\documentclass{article}
\usepackage{amsmath}
\usepackage{graphicx}
\usepackage{amsfonts}
\usepackage{amssymb}
\usepackage{alltt}
\usepackage{stmaryrd}
\usepackage{dsfont} 
\usepackage{color} 
\usepackage[standard]{ntheorem}
\usepackage{floatflt}
\usepackage[width=18cm,height=25.5cm]{geometry}

\begin{document}

\title{Chaotic iterations and topological chaos}
\date{}
\author{Jacques M. Bahi, Christophe Guyeux\\Laboratoire d'Informatique de l'universit\'{e} de Franche-Comt\'{e}, \\90000 Belfort cedex\\T\'{e}l: 03 84 58 77 94; fax: 03 84 58 77 32\\e-mail: jacques.bahi@univ-fcomte.fr, christophe.guyeux@univ-fcomte.fr}
\maketitle

\begin{center}
\textbf{Abstract}
\end{center}

Chaotic iterations have been introduced on the one hand by Chazan,
Miranker~\cite{Chazan69} and Miellou~\cite{Miellou75} in a numerical analysis context, and on the other hand by Robert~\cite{Robert1986} and Pellegrin~\cite{Pellegrin1986} in the discrete dynamical systems framework.
In both cases, the objective was to derive conditions of convergence of
such iterations to a fixed state. In this paper, a new point of view is
presented, the goal here is to derive conditions under which chaotic
iterations admit a chaotic behaviour in a rigorous mathematical sense.
Contrary to what has been studied in the literature, convergence is not
desired.\medskip 

More precisely, we establish in this paper a link between the concept of
chaotic iterations on a finite set and the notion of topological chaos~\cite%
{Li75},~\cite{Dev89},~\cite{Knudsen1994a}. We are motivated by concrete
applications of our approach, such as the use of chaotic boolean iterations
in the computer security field. Indeed, the concept of chaos is used in many
areas of data security without real rigorous theoretical foundations, and
without using the fundamental properties that allow chaos. The wish of this paper is to bring a bit more mathematical rigour in this field.

\section{Introduction}

Let us consider the \emph{system} $\mathds{B}^2 = \{0;1\}^2$, in which each of the
two \emph{cells} $c_i$ is caracterized by a boolean state $e_i$. An \emph{evolution rule} is, for example,

\[
\begin{array}{rccl}
f: & \mathds{B}^2 & \longrightarrow & \mathds{B}^2 \\ 
& (e_1,e_2) & \longmapsto & (e_1+\overline{e_2}, \overline{e_1}) \\ 
&  &  & 
\end{array}%
\]

These cells can be updated in a serial mode (the elements are iterated in a sequential mode, at each time only one
element is iterated), in a parallel mode (at each time, all the elements are
iterated), or by following a sequence $(S^n)_{n \in \mathds{N}}$: the $n^{th}$ term $S^n$ is constituted by the
block components to be updated at the $n^{th}$ iteration. This is the \emph{%
chaotic iterations}, and $S$ is called the \emph{strategy}.
Let us notice that serial and parallel modes are particular cases of chaotic iterations. Until now, only the conditions of convergence have been studied.

\medskip

A priori, the \emph{chaotic} adjective means ``in a disorder
way'', and has nothing to do with the mathematical theory of chaos, studied
by Li-Yorke~\cite{Li75}, Devaney~\cite{Dev89}, Knudsen~\cite{Knudsen1994a}, 
\emph{etc.} We asked ourselves what it really was.

In this paper we study the topological evolution of a system during chaotic
iterations. To do so, chaotic iterations have been written in the field of
discrete dynamical system:

\[
\left\{%
\begin{array}{l}
x^0 \in \mathcal{X} \\ 
x^{n+1} = f(x^n)%
\end{array}
\right. 
\]

\noindent where $(\mathcal{X},d)$ is a metric space (for a distance to be
defined), and $f$ is continuous.

\medskip

Thus, it becomes possible to study the topology of chaotic iterations. More
exactly, the question: ``Are the chaotic iterations a topological chaos ?''
has been raised.%

This study is the first of a series we intend to carry out. We think that
the mathematical framework in which we are placed offers interesting new
tools allowing the conception, the comparison and the evaluation of new
algorithms where disorder, hazard or unpredictability are to be considered.

\bigskip

The rest of the paper is organised as follows.\newline
The first next section is devoted to some recalls on the domain of topological chaos and the domain of discrete chaotic iterations. Third and fourth sections constitute the theoretical study of the present paper.
In section 6, the computer and so the finite set of machine numbers is considered. The paper ends with some discussions and future work.

\section{Basic recalls}

This section is devoted to basic notations and terminologies in
the fields of topological chaos and chaotic iterations.

\subsection{Chaotic iterations}

In the sequel $S^{n}$ denotes the $n^{th}$ term of a sequence $S$, $V_{i}$
denotes the $i^{th}$ component of a vector $V$, and $f^{k}=f\circ ...\circ f$
denotes the $k^{th}$ composition of a function $f$. Finally, the following
notation is used: $\llbracket1;N\rrbracket=\{1,2,\hdots,N\}$.

\medskip Let us consider a \emph{system} of a finite number $\mathsf{N}$ of 
\emph{cells}, so that each cell has a boolean \emph{state}. Then a sequence
of length $\mathsf{N}$ of boolean states of the cells corresponds to a
particular \emph{state of the system}.

A \emph{strategy} corresponds to a sequence $S$ of $\llbracket1;\mathsf{N}%
\rrbracket$. The set of all strategies is denoted by $\mathbb{S}.$

\begin{definition}
Let $S\in \mathbb{S}$. The \emph{shift} function is defined by $\sigma :
(S^n)_{n\in \mathds{N}} \in \mathbb{S} \longrightarrow (S^{n+1})_{n\in %
\mathds{N}} \in \mathbb{S}$, and the \emph{initial function} $i$ is the map
which associates to a sequence, its first term: $i: (S^n)_{n\in \mathds{N}%
}\in \mathbb{S} \longrightarrow S^0 \in \llbracket1;\mathsf{N}\rrbracket $.
\end{definition}

The set $\mathds{B}$ denoting $\{0,1\}$, let $f:\mathds{B}^{\mathsf{N}%
}\longrightarrow \mathds{B}^{\mathsf{N}}$ be a function, and $S\in \mathbb{S}
$ be a strategy. Then, the so called \emph{chaotic iterations} are defined by

\begin{equation}
\left. 
\begin{array}{l}
x^0\in \mathds{B}^{\mathsf{N}}, \\ 
\forall n\in \mathds{N}^{\ast },\forall i\in \llbracket1;\mathsf{N}\rrbracket%
,x_i^n=\left\{ 
\begin{array}{ll}
x_i^{n-1} & \text{ if }S^n\neq i \\ 
\left(f(x^n)\right)_{S^n} & \text{ if }S^n=i.%
\end{array}%
\right.%
\end{array}%
\right.  \label{chaotic iterations}
\end{equation}

In other words, at the $n^{th}$ iteration, only the $S^{n}-$th
cell is \textquotedblleft iterated\textquotedblright . Note that in a
more general formulation, $S^n$ can be a subset of components, and $f(x^{n})_{S^{n}}$ can be replaced by $f(x^{k})_{S^{n}}$, where $k\leqslant n$, modelizing for example delay
transmission (see \emph{e.g.}~\cite{Bahi2000}). For the general definition
of such chaotic iterations, see, e.g.~\cite{Robert1986}.

\subsection{Devaney's chaotic dynamical systems}

Consider a metric space $(\mathcal{X},d)$, and a continuous function $f:%
\mathcal{X}\longrightarrow \mathcal{X}$.

\begin{definition}
$f$ is said to be \emph{topologically transitive} if, for any pair of open
sets $U,V \subset \mathcal{X}$, there exists $k>0$ such that $f^k(U) \cap V
\neq \varnothing$.
\end{definition}

\begin{definition}
$(\mathcal{X},f)$ is said to be \emph{regular} if the set of periodic points
is dense in $\mathcal{X}$.
\end{definition}

\begin{definition}
\label{sensitivity} $f$ has \emph{sensitive dependence on initial conditions}
if there exists $\delta >0$ such that, for any $x\in \mathcal{X}$ and any
neighbourhood $V$ of $x$, there exists $y\in V$ and $n\geqslant 0$ such that $%
|f^{n}(x)-f^{n}(y)|>\delta $.

$\delta$ is called the \emph{constant of sensitivity} of $f$.
\end{definition}

Let us now recall the definition of a chaotic topological system, in the
sense of Devaney~\cite{Dev89}:

\begin{definition}
$f:\mathcal{X}\longrightarrow \mathcal{X}$ is said to be \emph{chaotic} on $%
\mathcal{X}$ if $(\mathcal{X},f)$ is regular, topologically transitive, and
has sensitive dependence on initial conditions.
\end{definition}

\section{A topological approach for chaotic iterations}

In this section we will put our study in a topological context by defining a
suitable metric set.

\subsection{The iteration function and the phase space}

\label{Defining}

Denote by $\delta $ the \emph{discrete boolean metric}, $\delta
(x,y)=0\Leftrightarrow x=y.$ Given a function $f$, \ define the function

\[
\begin{array}{lrll}
F_{f}: & \llbracket1;\mathsf{N}\rrbracket\times \mathds{B}^{\mathsf{N}} & 
\longrightarrow  & \mathds{B}^{\mathsf{N}} \\ 
& (k,E) & \longmapsto  & \left( E_{j}.\delta (k,j)+f(E)_{k}.\overline{\delta
(k,j)}\right) _{j\in \llbracket1;\mathsf{N}\rrbracket},%
\end{array}%
\]%
where + and . are boolean operations.

Consider the phase space: $\mathcal{X}=\llbracket1;\mathsf{N}\rrbracket^{%
\mathds{N}}\times \mathds{B}^{\mathsf{N}}$, which has the cardinality of the
continuum, and the map 
\[
G_{f}\left( S,E\right) =\left( \sigma (S),F_{f}(i(S),E)\right) \label{Gf},
\]
then the chaotic iterations defined in (\ref{chaotic iterations}) can be
described by the following iterations 
\[
\left\{ 
\begin{array}{l}
X^{0}\in \mathcal{X} \\ 
X^{k+1}=G_{f}(X^{k}).%
\end{array}%
\right. 
\]

\subsection{A new distance}

We define a new distance between two points $(S,E),(\check{S},\check{E})\in 
\mathcal{X}$ by $d((S,E);(\check{S},\check{E}))=d_{e}(E,\check{E})+d_{s}(S,%
\check{S})$, where $$\displaystyle{d_{e}(E,\check{E})} = \displaystyle{%
\sum_{k=1}^{\mathsf{N}}\delta (E_{k},\check{E}_{k})} \textrm{ and }\displaystyle{%
d_{s}(S,\check{S})} = \displaystyle{\dfrac{9}{\mathsf{N}}\sum_{k=1}^{\infty }%
\dfrac{|S^k-\check{S}^k|}{10^{k}}}.$$

\subsection{The topological framework}

It can be easily proved that,

\begin{theorem}
\label{continuite} $G_f$ is continuous on $(\mathcal{X},d)$.
\end{theorem}

Then chaotic iterations can be seen as a dynamical system in a topological
space. In the next section, we will show that chaotic iterations are a case
of topological chaos, in the sense of Devaney~\cite{Dev89}.

\section{Discrete chaotic iterations and topological chaos}

To prove that we are in the framework of Devaney's topological chaos, we
have to check the regularity and transitivity conditions.

\subsection{Regularity}

\label{regularite}

\begin{theorem}
Periodic points of $G_{f}$ are dense in $\mathcal{X}$.
\end{theorem}

\begin{proof}
Let $(S,E)\in \mathcal{X}$, and $\varepsilon >0$. We are looking for a
periodic point $(S^{\prime },E^{\prime })$ satisfying $d((S,E);(S^{\prime
},E^{\prime }))<\varepsilon$.

We choose $E^{\prime }=E$, and we reproduce enough entries from $S$ to $%
S^{\prime }$ so that the distance between $(S^{\prime },E)$ and $(S,E)$ is
strictly less than $\varepsilon $: a number $k=\lfloor log_{10}(\varepsilon
)\rfloor +1$ of terms is sufficient.\newline
After this $k^{th}$ iterations, the new common state is $\mathcal{E}$, and
strategy $S^{\prime }$ is shifted of $k$ positions: $\sigma ^{k}(S^{\prime })$.\newline
Then we have to complete strategy $S^{\prime }$ in order to make $(E^{\prime
},S^{\prime })$ periodic (at least for sufficiently large indices). To do
so, we put an infinite number of 1 to the strategy $S^{\prime }$.

Then, either the first state is conserved after one iteration, so $\mathcal{E%
}$ is unchanged and we obtain a fixed point. Or the first state is not
conserved, then: if the first state is not conserved after a second
iteration, then we will be again in the first case above (due to the fact that a state is a boolean). Otherwise the first state is conserved, and we have
indeed a fixed (periodic) point.

Thus, there exists a periodic point into every neighbourhood of any point, so 
$(\mathcal{X},G_f)$ is regular, for any map $f$.
\end{proof}

\subsection{Transitivity}

\label{transitivite} Contrary to the regularity, the topological
transitivity condition is not automatically satisfied by any function \newline ($
f=Identity$ is not topologically transitive). Let us denote by $\mathcal{T}$
the set of maps $f$ such that $(\mathcal{X},G_{f})$ is topologically
transitive.

\begin{theorem}
$\mathcal{T}$ is a nonempty set.
\end{theorem}

\begin{proof}
We will prove that the vectorial logical negation function $f_{0}$

\begin{equation}
\begin{array}{rccc}
f_{0}: & \mathds{B}^{\mathsf{N}} & \longrightarrow & \mathds{B}^{\mathsf{N}}
\\ 
& (x_{1},\hdots,x_{\mathsf{N}}) & \longmapsto & (\overline{x_{1}},\hdots,%
\overline{x_{\mathsf{N}}}) \\ 
&  &  & 
\end{array}
\label{f0}
\end{equation}%
\noindent is topologically transitive.\newline
Let $\mathcal{B}_A=\mathcal{B}(X_{A},r_{A})$ and $\mathcal{B}_B=\mathcal{B}(X_{B},r_{B})$ be two
open balls of $\mathcal{X}$, where $X_A=(S_A,E_A)$, and $X_B=(S_B,E_B)$. Our goal is to start from a point of $\mathcal{B}_A$ and to arrive, after some iterations of $G_{f_0}$, in $\mathcal{B}_B$.\newline
We have to be close to $X_{A}$, then the starting state $E$ must be $E_{A}$; it
remains to construct the strategy $S$. Let $S^n = S_{A}^n, \forall n \leqslant n_0$, where $n_0$ is chosen in such a way that $(S,E_{A})\in \mathcal{B}_{A}$, and $E'$ be the state of $G_{f_0}^{n_0}(S_{A},E_{A})$.\newline
$E'$ differs from $E_{B}$ by a finite number of cells $c_1,\hdots, c_{n_1}$. Let $S^{n_0+n} = c_n, \forall n \leqslant n_1$. Then the state of $G_{f_0}^{n_0+n_1}(S,E)$ is $E_{B}$.\newline
Last, let $S^{n_0+n_1+n} = S_{B}^n, \forall n \leqslant n_2$, where $n_2$ is chosen in such a way that $G_{f_0}^{n_0+n_1}(S,E)$ is at a distance less than $r_B $ from $(S_{B},E_{B})$. Then, starting from a point $(S,E)$ close to $X_A$, we are close to $X_B$ after $n_0+n_1$ iterations: $(\mathcal{X},G_{f_0})$ is transitive.
\end{proof}

\subsection{Sensitive dependence on initial conditions}

\label{sensibilite}

\begin{theorem}
$(\mathcal{X},G_{f_0})$ has sensitive dependence on initial conditions, and
its constant of sensitiveness is equal to $\mathsf{N}$.
\end{theorem}

\begin{Proof}
Let $(S,E) \in \mathcal{X}$, and $\delta>0$. A new point $(S',E')$ is defined by: $E'=E$, $S'^n = S^n, \forall n \leqslant n_0$, where $n_0$ is chosen in such a way that $d((S,E);(S',E'))<\delta$, and $S'^{n_0+k} = k, \forall k \in \llbracket 1; \mathsf{N} \rrbracket$.

\noindent Then the point $(S',E')$ is as close as we want than $(S,E)$, and systems of $G_{f_0}^{k+\mathsf{N}}(S,E)$ and $G_{f_0}^{k+\mathsf{N}}(S',E')$ have no cell presenting the same state: distance between this two points is greater or equal than $\mathsf{N}$.
\end{Proof}

\begin{remark}
This sensitive dependence could be stated as a consequence of regularity and
transitivity (by using the theorem of Banks~\cite{Banks92}). However, we have
preferred proving this result independently of regularity, because the
notion of regularity must be redefined in the context of the finite set of
machine numbers (see section \ref{Concerning}).
\end{remark}

\subsection{Chaos}

In conclusion, if $f\in \mathcal{T}$, then $(\mathcal{X},G_{f})$ is
topologically transitive, regular and has sensitive dependence on initial
conditions. Then we have the result.

\begin{theorem}
$\forall f\in \mathcal{T}\neq \varnothing ,$ $G_{f}$ is a chaotic map on $(%
\mathcal{X},d)$ in the sense of Devaney.
\end{theorem}

We have proven that under the transitivity condition of $f$, chaotic
iterations generated by $f$ can be described by a chaotic map on a
topological space in the sense of Devaney. We have considered a finite set
of states  $\mathds{B}^{\mathsf{N}}$ and a set $\mathbb{S}$ of strategies
composed by an infinite number of infinite sequences. In the following
section we will discuss the impact of these assumptions in the context of the
finite set of machine numbers.

\section{The case of finite strategies}

\label{Concerning}


In the computer science framework, we also have to deal with a finite set of
states of the form $\mathds{B}^{\mathsf{N}}$ and the set $\mathbb{S}$ of
sequences of $\llbracket 1; \mathsf{N} \rrbracket$ is infinite (countable), so in
practice the set $\mathcal{X}$ is also infinite. The only difference with
respect to the theoretical study comes from the fact that the sequences of $%
\mathbb{S}$ are of finite but not fixed length in the practice.\newline
The proof of the continuity, the transitivity and the
sensitivity conditions are independent of the finitude of the length of
strategies (sequences of $\mathbb{S}$), so even in the case of finite machine numbers, we have the two fundamental properties of chaos: sensitivity and transitivity, which respectively implies unpredictability and indecomposability (see~\cite{Dev89}, p.50). The regularity property has no meaning in the case of finite systems because of the notion of periodicity. We propose a new definition  in order to bypass the notion of periodicity in practice.
\medskip

\begin{definition}
A strategy $S=(S^{1},..., S^{L})$ is said \emph{cyclic} if a subset of successive
terms is repeated from a given rank, until the end of $S$. A point of $\mathcal{X}$ that admits a cyclic strategy is called a \emph{cyclic point}.
\end{definition}
For example, $(1,3,2,4,1,2,1,2)$ and $(1,3,2,4,1,2,2,2)$ are cyclic, but $(1,3,2,4,1,2)$ and $(1,3,2,1,3)$ are not cyclic.
This definition can be interpreted as the analogous of periodicity on finite
sets.
Then, following the proof of regularity (section \ref{regularite}), it can be proved that
the set of cyclic points is dense on $\mathcal{X}$, hence obtaining a desired element of regularity in finite sets, as quoted by Devaney (\cite{Dev89}, p.50): two points arbitrary close to each other could have different behaviours, the one could have a cyclic behaviour as long as the system iterates while the trajectory of the second could "visit" the
whole phase space. It should be recalled that the regularity was introduced
by Devaney in order to counteract the transitivity and to obtain such a
property: two points close to each other can have fundamental different behaviours.

\bigskip

It is worthwhile to notice that even if the set of machine numbers is
finite, we deal with strategies that have a finite but unbounded length.
Indeed, it is not necessary to store all the terms of the strategy in
the memory, only the $n^{th}$ term (an integer less than or equal to $N$) of
the strategy has to be stored at the $n^{th}$ step, as it is illustrated in
the following example. Let us suppose that a given text is input from the outside world in the computer character by character, and that the current term of the strategy is given by
the ASCII code of the current stored character. Then, as the set of all possible texts of the  outside world is infinite and the number of their characters is
unbounded, we have to deal with an infinite set of finite but unbounded
strategies. Of course, the preceding example is a simplistic illustrating example.
A chaotic procedure should to be introduced to generate the terms of the
strategy from the stream of characters.\newline
In conclusion, even in the computer science framework our previous theory applies.

\section{Discussion and future work}

We proved that discrete chaotic iterations are a particular case of
Devaney's topological chaos if the iteration function is topologically
transitive, and that the set of topologically transitive functions is non
void. \newline
This theory has a lot of applications, because of the high number of
situations that can be described with the chaotic iterations: neural
networks, cellular automata, multi-processor computing, and more generally
any aggregation of cells (such as pictures, movies, sounds). If this system is requested to evolve in an apparently disorderly manner, \emph{e.g.} for security reasons (encryption, watermarking, pseudo-random number generation, hash functions, \emph{etc.}), our
results could be useful. Moreover, the theory brings another way to compare two given algorithms concerned by
disorder (evaluation of theirs constants of sensitivity, expansivity, \emph{etc.}), which
can be seen as a complement of existing statistical evaluations.\newline
In future work, other forms of chaos (such as Li-York chaos~\cite{Li75})
will be studied, other quantitative and qualitative tools such as
expansivity or entropy (see e.g.~\cite{Adler65} or~\cite{Bowen}) will be
explored, and the domain of applications of our theoretical concepts will be
enlarged.

\bibliographystyle{plain}
\bibliography{mabase.bib}

\newpage

\section{Appendix: Continuity of $G_f$}

\begin{theorem}
$G_f$ is a continuous function for $(\mathcal{X},d)$.
\end{theorem}

\begin{proof}
We use the sequential continuity (we are in a metric space).

Let $(S^n,E^n)_{n\in \mathds{N}}$ be a sequence of the phase space $\mathcal{%
X}$, which converges to $(S,E)$. We will prove that $\left(
G_{f}(S^n,E^n)\right) _{n\in \mathds{N}}$ converges to $G_{f}(S,E) $. Let us recall that for all $n$, $S^n$ is a strategy,
thus, we consider a sequence of strategy (\emph{i.e.} a sequence of
sequences).\newline
As $d((S^n,E^n);(S,E))$ converges to 0, each distance $d_{e}(E^n,E)$ and $d_{s}(S^n,S)$ converges to 0. But $d_{e}(E^n,E)$ is an integer, so $\exists n_{0}\in \mathds{N},$ $%
d_{e}(E^n,E)=0$ for any $n\geqslant n_{0}$.\newline
In other words, there exists threshold $n_{0}\in \mathds{N}$ after which no
cell will change its state: 
\[
\exists n_{0}\in \mathds{N},n\geqslant n_{0}\Longrightarrow E^n=E.
\]%
In addition, $d_{s}(S^n,S)\longrightarrow 0,$ so $\exists n_{1}\in \mathds{N}%
,d_{s}(S^n,S)<10^{-1}$ for all indices greater than or equal to $n_{1}$.
This means that for $n\geqslant n_{1}$, all the $S^n$ have the same first
term, which is $S_0$:%
\[
\forall n\geqslant n_{1},S_0^n=S_0.
\]%
Thus, after the $max(n_{0},n_{1})-$th term, states of $E^n$ and $E$ are the
same, and strategies $S^n$ and $S$ start with the same first term.\newline
Consequently, states of $G_{f}(S^n,E^n)$ and $G_{f}(S,E)$ are equal, then
distance $d$ between this two points is strictly less than 1 (after the rank 
$max(n_{0},n_{1})$).\bigskip \newline
\noindent We now prove that the distance between $\left(
G_{f}(S^n,E^n)\right) $ and $\left( G_{f}(S,E)\right) $ is convergent to 0.
Let $\varepsilon >0$. \medskip

\begin{itemize}
\item If $\varepsilon \geqslant 1$, then we have seen that the distance
between $\left( G_{f}(S^n,E^n)\right) $ and $\left( G_{f}(S,E)\right) $ is
strictly less than 1 after the $max(n_{0},n_{1})^{th}$ term (same state).
\medskip

\item If $\varepsilon <1$, then $\exists k\in \mathds{N},10^{-k}\geqslant
\varepsilon \geqslant 10^{-(k+1)}$. But $d_{s}(S^n,S)$ converges to 0, so 
\[
\exists n_{2}\in \mathds{N},\forall n\geqslant
n_{2},d_{s}(S^n,S)<10^{-(k+2)},
\]%
after $n_{2}$, the $k+2$ first terms of $S^n$ and $S$ are equal.
\end{itemize}

\noindent As a consequence, the $k+1$ first entries of the strategies of $%
G_{f}(S^n,E^n)$ and $G_{f}(S,E)$ are the same (because $G_{f}$ is a shift of
strategies), and due to the definition of $d_{s}$, the floating part of the
distance between $(S^n,E^n)$ and $(S,E)$ is strictly less than $%
10^{-(k+1)}\leqslant \varepsilon $.\bigskip \newline
In conclusion, $G_{f}$ is continuous,%
\[
\forall \varepsilon >0,\exists N_{0}=max(n_{0},n_{1},n_{2})\in \mathds{N}%
,\forall n\geqslant N_{0},d\left( G_{f}(S^n,E^n);G_{f}(S,E)\right) \leqslant
\varepsilon .
\]
\end{proof}

\end{document}